\documentclass[a4paper,12pt,twocolumn]{article}
\usepackage[utf8]{inputenc}
\usepackage{hyperref}
\usepackage{authblk}
\usepackage{fullpage}
\usepackage{graphicx}
\usepackage{amsmath, amssymb, amsthm}
\usepackage[alphabetic,nobysame, initials]{amsrefs}

\newtheorem{theorem}{Theorem}[section]
\newtheorem{lemma}[theorem]{Lemma}

\newtheorem{rem}[theorem]{Remark}

\renewcommand{\Re}{\mathbb R}
\renewcommand{\epsilon}{\varepsilon}
\newcommand{\Red}{\Re^d}

\renewcommand{\phi}{\varphi}

\newcommand{\noshow}[1]{}

\newcommand{\cardin}[1]{\lvert {#1} \rvert}
\newcommand{\D}{\mathcal{D}}

\title{Bounding a global red-blue proportion using local conditions}
\author[1]{M\'arton Nasz\'odi \thanks{Email address: marton.naszodi@math.elte.hu.}}
\author[2]{Leonardo Mart\'inez-Sandoval \thanks{Email address: leomtz@im.unam.mx. }}
\author[3]{Shakhar Smorodinsky\thanks{Email address: shakhar@math.bgu.ac.il.}}
\affil[1]{\footnotesize{Department of Geometry, Lorand E\"otv\"os University, Budapest, Hungary}}
\affil[1]{EPFL, Lausanne, Switzerland}
\affil[2]{Department of Computer Science, Ben-Gurion University of the Negev, Be'er-Sheva Israel.}
\affil[3]{Department of Mathematics, Ben-Gurion University of the Negev, Be'er-Sheva Israel.}

\begin{document}
\maketitle

\begin{abstract}
We study the following local-to-global phenomenon: Let $B$ and $R$ be
two finite sets of (blue and red) points in the Euclidean plane $\Re^2$. Suppose
that in each ``neighborhood"  of a red point, the number of blue points is at
least as large as the number of red points. We show that in this case the total
number of blue points is at least one fifth of the total number of red
points. We also show that this bound is optimal and we generalize the result
to arbitrary dimension and arbitrary norm using results from Minkowski arrangements.
\end{abstract}


\section{Introduction}

Consider the following scenario in wireless networks. Suppose we have $n$ clients and $m$ antennas where both are represented as points in the plane (see Figure \ref{fig:hypothesis}). Each client has a wireless device that can communicate with the antennas. Assume also that each client is associated with some disk centered at the client's location and having radius representing how far in the plane his device can communicate. Suppose also, that some communication protocol requires that in each of the clients disks, the number of antennas is at least some fixed proportion $\lambda>0$ of the number of clients in the disk.  Our question is: does such a local requirement imply a global lower bound on the number of antennas in terms of the number of clients? In this paper we answer this question and provide exact bounds. Let us formulate the problem more precisely.

\begin{figure}
	\centering
		\includegraphics[width=0.40\textwidth]{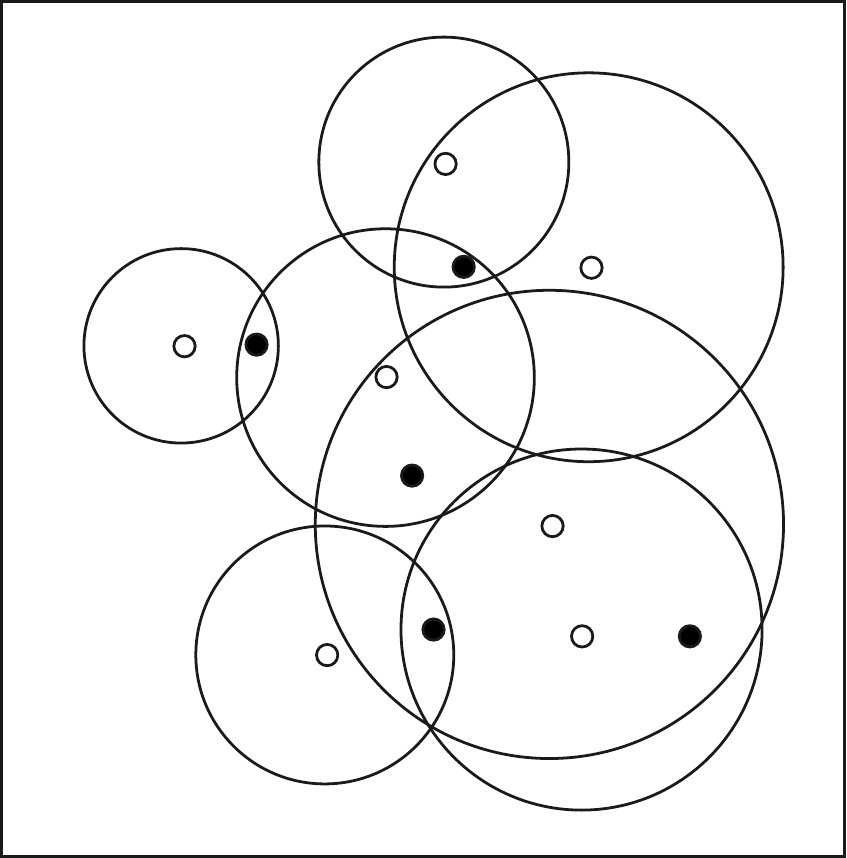}
	\caption{In each device range (each disk) there are at least as many antennas (black dots) as devices (white dots), so the hypothesis holds for $\lambda=1$.}
	\label{fig:hypothesis}
\end{figure}

Let $B$ and $R=\{p_1,\ldots,p_n\}$ be two finite sets in $\Re^2$.
Let $\D = \{D_1,\ldots,D_n\}$ be a set of Euclidean disks centered at the red
points,
i.e., the center of $D_i$ is $p_i$. Let $\{\rho_1,\ldots,\rho_n\}$ be the radii
of the disks in $\D$.

\begin{theorem}\label{thm:euclideanplane}
Assume that for each $i$ we have $\cardin{D_i \cap B} \geq \cardin{D_i \cap
R}$. Then
$\cardin{B} \geq \frac{n}{5}$.
Furthermore, the multiplicative constant $\frac{1}{5}$ cannot be improved.
\end{theorem}

%


Such a local-to-global ratio phenomenon was shown to be useful in a more combinatorial setting. Pach et. al. \cite{Pach2015},  solved a conjecture by Richter and Thomassen \cite{Richter1995} on the number of total ``crossings" that a family of pairwise intersecting curves in the plane in general position can have. Lemma 1 from their paper is a first step in the proof and it consists of a local-to-global phenomenon as described above.


We will obtain Theorem \ref{thm:euclideanplane} from a more general result. In order to state it,
we introduce some terminology.

Let $K$ be an origin-symmetric convex body in $\Red$, that is, the unit ball of
a norm.

A \emph{strict Minkowski arrangement} is a family $\D=\{K_1=p_1+\rho_1K,\ldots,
K_n=p_n+\rho_nK\}$ of homothets of $K$, where $p_i\in\Red$ and $\rho_i>0$, such
that no member of the family contains the center of another member. An
\emph{intersecting family} is a family of sets that all share at least one
element.

We denote the \emph{maximum cardinality of an intersecting strict Minkowski
arrangement} of homothets of $K$ by $M(K)$. It is known that $M(K)$ exists for every $K$ and $M(K)\leq 3^d$ (see, e.g., Lemma~21 of \cite{NPS16}).
On the other hand (somewhat surprisingly), there is an origin-symmetric convex body
$K$ in $\Re^d$ such that $M(K)=\Omega\left(\sqrt{7}^d\right)$, \cites{T98,NPS16}. For more on Minkowski arrangements
see, e.g., \cites{FL94}.

We need the following auxiliary Lemma.

\begin{lemma}\label{lemma:minkowski}
Let $K$ be an origin-symmetric convex body in $\Red$. Let $R=\{p_1,\ldots,p_n\}$ be a set of points in $\Red$ and let
$\D=\{K_1=p_1+\rho_1K,\ldots, K_n=p_n+\rho_nK\}$ be a family of homothets of
$K$. Then there exists a subfamily $\D' \subset \D$ that covers $R$ and forms a strict Minkowski arrangement.
Moreover, $\D^\prime$ can be found using a greedy algorithm.
\end{lemma}

As a corollary, we will obtain the following theorem.

\begin{theorem}\label{thm:minkowskispecific}
Let $K$ be an origin-symmetric convex body in $\Red$. Let $R=\{p_1,\ldots,p_n\}$ be a set of points in $\Red$ and let
$\D=\{K_1=p_1+\rho_1K,\ldots, K_n=p_n+\rho_nK\}$ be a family of homothets of
$K$ where $\rho_1,\ldots,\rho_n>0$. Let $B$ be another set of points
in $\Red$, and assume that, for some $\lambda>0$, we have
\begin{equation}
\frac{\cardin{B\cap K_i}}{\cardin{R\cap K_i}} \geq \lambda,
\end{equation}
for all $i\in[n]$.
Then $\frac{\cardin{B}}{\cardin{R}} \geq \frac{\lambda}{3^d}$.
\end{theorem}

In Theorem \ref{thm:euclideanplane} the convex body $K$ is a Euclidean unit disk in the plane. Another case of special interest is when the convex body $K$ is a unit cube and thus it induces the $\ell_\infty$ norm. In this situation we get a sharper and optimal inequality. 
\begin{theorem}
\label{thm:cubes}
	If $K$ is the unit cube in $\mathbb{R}^d$, then the conclusion in Theorem \ref{thm:minkowskispecific} can be strengthened to $\frac{|B|}{|R|}\geq \frac{\lambda}{2^d}$. Furthermore, the multiplicative constant $\frac{1}{2^d}$ cannot be improved.
\end{theorem}

In the results above, the points $p_i$ play the role of the centers of the sets of the Minkowski arrangement. One might ask if this restriction is essential. As a final result, we give a general construction to show that it is.

\begin{theorem}\label{thm:badexample}
		Let $K$ be any convex body in the plane and $\epsilon,\lambda$ any positive real numbers. There exist sets of points $R=\{p_1,\ldots,p_n\}$ and $B$ in the plane such that $|B|<\epsilon n$ and that for each $i$ there is a translate $K_i$ of $K$ that contains $p_i$ for which $|B\cap K_i|\geq \lambda |R\cap K_i|$.
\end{theorem}

In particular, even if each red point is contained in a unit disk with many blue points, the global blue to red ratio can be as small as desired. This is a possibly counter-intuitive fact in view of Theorem \ref{thm:euclideanplane}.




\section{Proofs}

\begin{proof}[Proof of Lemma~\ref{lemma:minkowski}]
  We construct a subfamily $\D^{\prime}$ of $\D$ with the property that no
member of $\D^{\prime}$ contains the center of any member of $\D^{\prime}$,
and $\bigcup \D^{\prime}$ covers the red points, $R$.
Assume without loss of generality that the labels of the points in $R$ are
sorted in non-increasing order of the homothety ratio, that
is, $\rho_1 \geq \cdots \geq \rho_n$. See Figure \ref{fig:cover} for an example.

\begin{figure}
	\centering
		\includegraphics[width=0.40\textwidth]{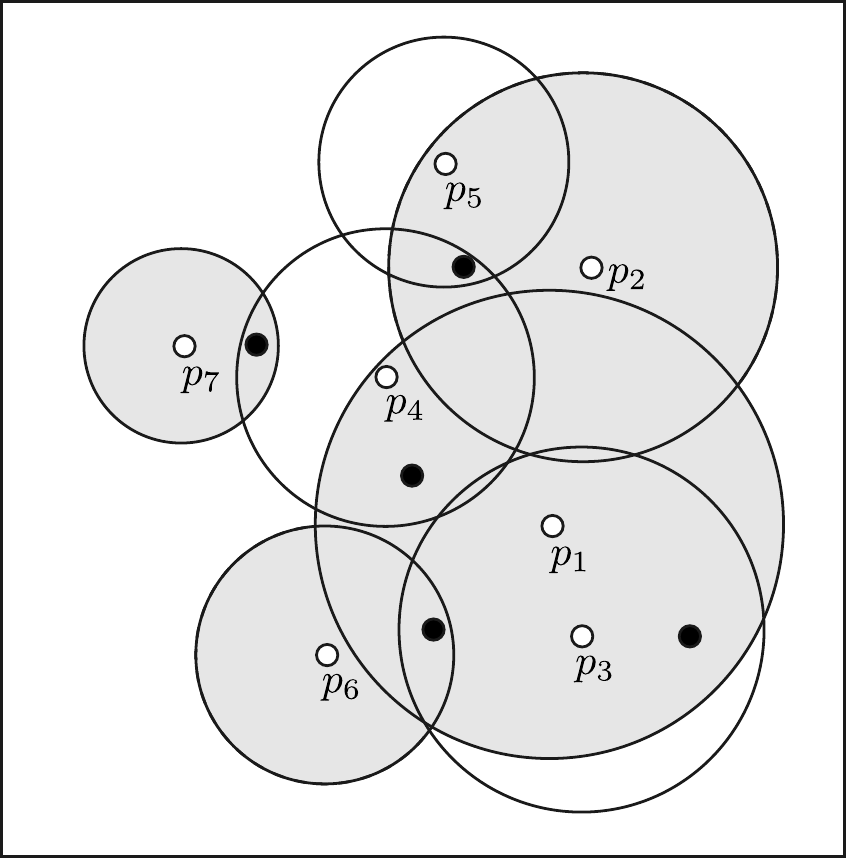}
	\caption{The centers of the disks are labeled in decreasing order of corresponding radii. The shaded disks cover the white points and no shaded disk contains the center of another.}
	\label{fig:cover}
\end{figure}

  We construct $\D^{\prime}$ in a greedy manner as follows: Add $K_1$ to
$\D^{\prime}$. Among
all red points that are not already covered by $\D^{\prime}$ pick a point $p_j$ whose
corresponding homothet $K_j$ has maximum homothety ratio $\rho_j$. Add $K_j$ to $\D^{\prime}$
and repeat until all red points are covered by $\D^{\prime}$. Note that the
homothets in $\D^{\prime}$ are not necessarily disjoint.

Clearly, $R\subset \bigcup \D^{\prime}$. Now we show that no member of $\D^{\prime}$
contains the center of another. Suppose to the contrary that $K_i$ contains the center
of $K_j$. If $i<j$, then $\rho_i\geq \rho_j$ so $K_i$ was chosen first, a contradiction to the fact that  $p_j$ was chosen among the points not covered by previous homothets. If $i>j$, then $K_j$ also contains the center of $K_i$, and we get a similar contradiction.

This finishes the proof of Lemma~\ref{lemma:minkowski}.

\end{proof}

\begin{proof}[Proof of Theorem~\ref{thm:minkowskispecific}]
By Lemma~\ref{lemma:minkowski}, there exists a subfamily $\D' \subset \D$ that covers $R$ and form a strict Minkowski arrangement.
Namely, $\bigcup \D^{\prime}$ covers $R$, and no point of $B$ is
contained in more than $M(K)$ members of $\D^{\prime}$.
In particular, it follows that
$$
\cardin{R} \leq \sum_{K \in \D'}\cardin{R\cap K} \leq  \sum_{K \in \D'}\frac{\cardin{B \cap K}}{\lambda}\leq \frac{M(K)}{\lambda}\cardin{B}
$$

so $$\frac{\cardin{B}}{\cardin{R}} \geq \frac{\lambda}{M(K)} \geq \frac{\lambda}{3^d}.$$ This completes the proof.
\end{proof}

\begin{lemma}\label{lem:euclideanMbound}
Let $K$ be the Euclidean unit disk centered at the origin. Then $M(K)=5$.
\end{lemma}
\begin{proof}[Proof of Lemma~\ref{lem:euclideanMbound}]
Five unit disks centered in the vertices of a unit-radius regular pentagon show that $M(K)\geq 5$. See Figure \ref{fig:optimal}a.

\begin{figure}
	\centering
		\includegraphics[width=0.30\textwidth]{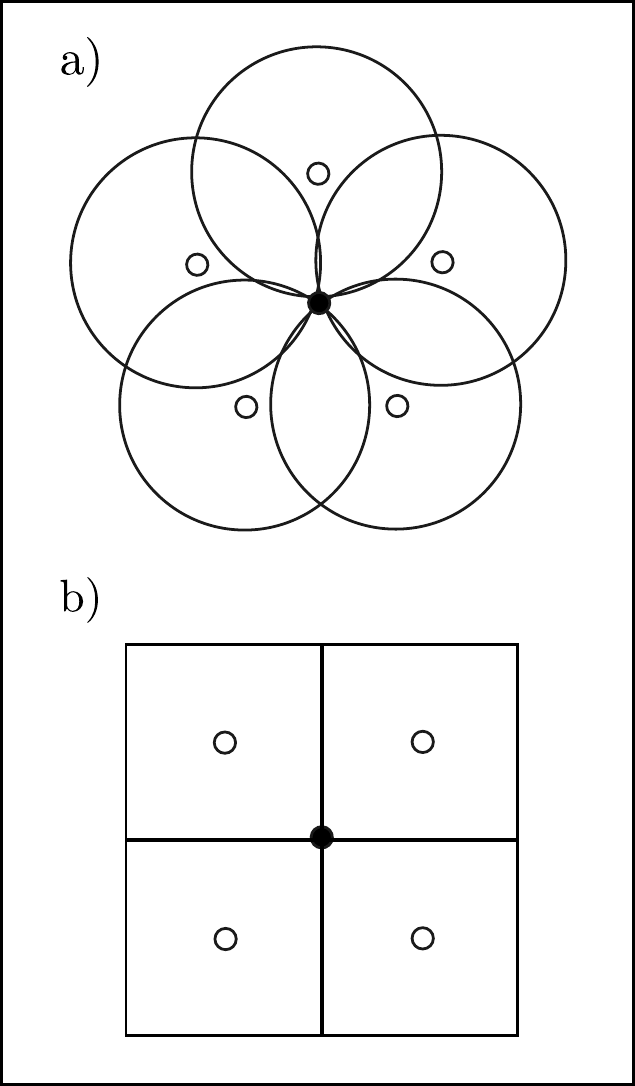}
	\caption{Optimal Minkowski arrangements in the plane for a) Euclidean disks, b) axis-parallel squares.}
	\label{fig:optimal}
\end{figure}

To prove the other direction, suppose that
there is a point $b$ in the plane that is contained in $6$ Euclidean disks in a
strict Minkowski arrangement. Then, by the pigeonhole principle, there are two
centers of those disks, say $p$ and $q$ such that the angle
$\sphericalangle(pbq)$ is at most
$60^\circ$.
Assume without loss of generality that $pb \geq qb$.
It is easily verified e.g., by the law of cosines, that the distance $pq$ is
less than $pb$. Hence, the disk centered at $p$ contains $q$, a contradiction. This completes the proof.
\end{proof}

\begin{lemma}\label{lem:cubebound}
Let $K$ be the unit cube of $\mathbb{R}^d$ centered at the origin. Then $M(K)=2^d$.
\end{lemma}

\begin{proof}[Proof of Lemma~\ref{lem:cubebound}] Let $d$ be a positive integer and $e_1,e_2,\ldots,e_n$ the canonical base of $\mathbb{R}^d$. Consider all the cubes of radius $1$ centered at each point of the form $\pm e_1 \pm e_2 \pm \ldots \pm e_d$. This family shows that $M(K)\geq 2^d$. See Figure \ref{fig:optimal}b for an example on the plane.

Now we show the other direction. Consider the $2^d$ closed regions of $\mathbb{R}^d$ bounded by the hyperplanes $x_i=0$ $i=1,2,\ldots,d$ and suppose on the contrary that we have an example with $2^d+1$ cubes or more that contain the origin. By the pidgeon-hole principle there is a region with at least two cube centers $u$ and $v$. By applying a rotation we may assume that it is the region of vectors with non-negative entries. We may also assume $\delta:=\|u\|_\infty \geq \|v\|_\infty$.

Since the $d$-cube centered at $u$ contains the origin, its radius must be at least $\delta$. We claim that this cube contains $v$. Indeed, each of the entries of $u$ and $v$ are in the interval $[0,\delta]$. So each of the entries of $u-v$ are in $[-\delta,\delta]$. Then $\|u-v\|_\infty \leq \delta$ as claimed. This contradiction finishes the proof.
\end{proof}

Theorem~\ref{thm:euclideanplane} clearly follows from combining the proof of
Theorem~\ref{thm:minkowskispecific} (with $\lambda=1$) and
Lemma~\ref{lem:euclideanMbound}. The result is sharp because we have equality when $R$ is the set of vertices of a regular pentagon with center $p$ and $B=\{p\}$. Similarly, Theorem \ref{thm:cubes} and its optimality follow from Lemma \ref{lem:cubebound}.

\begin{rem}

Lemma~\ref{lem:euclideanMbound} can be generalized to arbitrary dimension. This
implies
that Theorem~\ref{thm:euclideanplane} can be generalized to arbitrary dimension
almost verbatim.
\end{rem}

\begin{proof}[Proof of Theorem~\ref{thm:badexample}] Let $K$ be any convex body in the plane. We construct sets $R$ and $B$ as follows. Let $\ell$ be a tangent line of $K$ which intersects $K$ at exactly one point $t$. Let $I$ be a non-degenerate closed line segment contained in $K$ and parallel to $\ell$. Let $J$ be the (closed) segment that is the locus of the point $t$ as $K$ varies through all its translations in direction $d$ that contain $I$. See Figure \ref{fig:intervals}.

\begin{figure}
	\centering
		\includegraphics[width=0.40\textwidth]{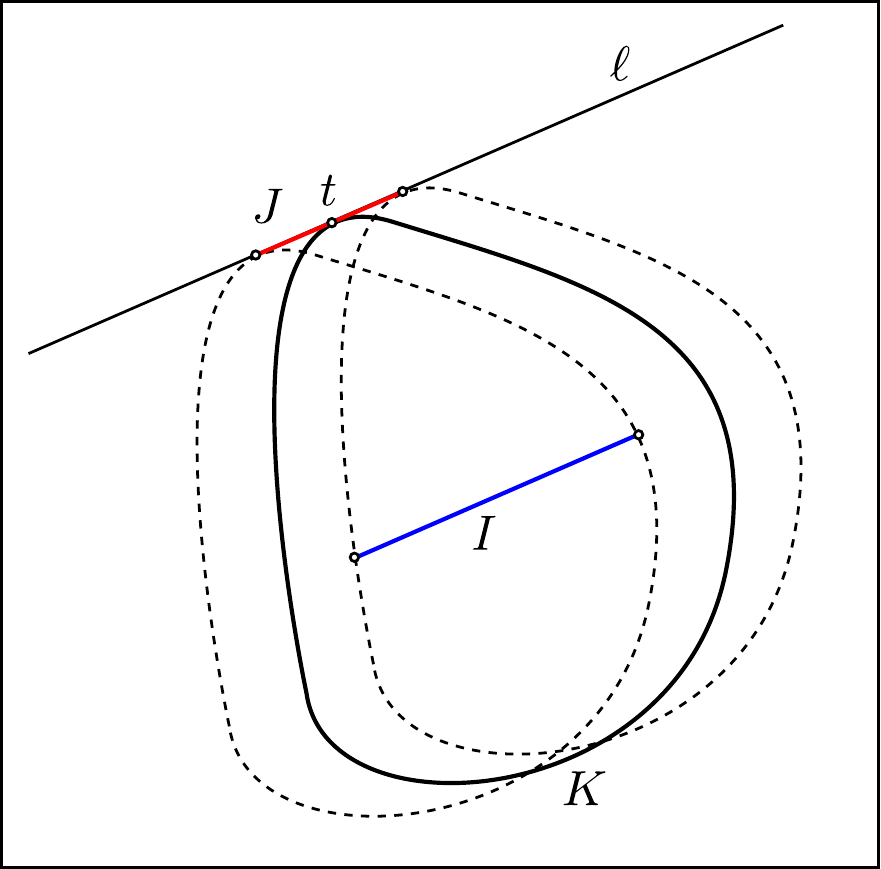}
	\caption{Construction of example without local-to-global phenomenon.}
	\label{fig:intervals}
\end{figure}

We construct $R$ by taking any $n$ points from $J$ and we construct $B$ by taking any $m$ points from $I$. For any point in $R$ there is a translation of $K$ that contains exactly one point of $R$ and $m$ points of $B$, which makes the local $B$ to $R$ ratio equal to $m$. But globally we can make the ratio $\frac{m}{n}$ arbitrarily small.
\end{proof}

\section*{Acknowledgements}

M. Nasz\'odi acknowledges the support of the J\'anos Bolyai Research Scholarship of the Hungarian Academy of Sciences, and
the National Research, Development, and Innovation Office, NKFIH Grant PD-104744, as well as the support of the Swiss National Science
Foundation grants 200020-144531 and 200020-162884.

L. Martinez-Sandoval's research was partially carried out during the author's visit at EPFL. The project leading to this application has received funding from the European Research Council (ERC) under the European Union’s Horizon 2020 research and innovation programme grant No. 678765 and from the Israel Science Foundation grant No. 1452/15.

S. Smorodinsky's research was partially supported by Grant 635/16 from the Israel Science Foundation. A part of this research was carried out
during the author's visit at EPFL, supported by Swiss National Science Foundation grants 200020-162884 and 200021-165977.

\bibliographystyle{plain}
\bibliography{biblio}

\begin{bibdiv}
\begin{biblist}

\bib{FL94}{article}{
      author={F{\"u}redi, Zolt{\'a}n},
      author={Loeb, Peter~A.},
       title={On the best constant for the {B}esicovitch covering theorem},
        date={1994},
        ISSN={0002-9939},
     journal={Proc. Amer. Math. Soc.},
      volume={121},
      number={4},
       pages={1063\ndash 1073},
         url={http://dx.doi.org/10.2307/2161215},
      review={\MR{1249875 (95b:28003)}},
}

\bib{NPS16}{article}{
      author={Naszódi, Márton},
      author={Pach, János},
      author={Swanepoel, Konrad},
       title={Arrangements of homothets of a convex body},
        date={2016},
     journal={arXiv:1608.04639 [math]},
         url={http://arxiv.org/abs/1608.04639},
        note={arXiv: 1608.04639},
}

\bib{Pach2015}{article}{
      author={{Pach}, J.},
      author={{Rubin}, N.},
      author={{Tardos}, G.},
       title={Beyond the {R}ichter-{T}homassen conjecture},
        date={2015},
     journal={arXiv:1504.08250 [math]},
         url={https://arxiv.org/abs/1504.08250},
        note={arXiv: 1504.08250},
}

\bib{Richter1995}{article}{
      author={Richter, R.~B.},
      author={Thomassen, C.},
       title={Intersections of curve systems and the crossing number of ${C}_5
  \times {C}_5$},
        date={1995},
        ISSN={1432-0444},
     journal={Discrete {\&} Computational Geometry},
      volume={13},
      number={2},
       pages={149\ndash 159},
         url={http://dx.doi.org/10.1007/BF02574034},
}

\bib{T98}{article}{
      author={Talata, I.},
       title={Exponential lower bound for the translative kissing numbers of
  $d$-dimensional convex bodies},
        date={1998},
        ISSN={0179-5376},
     journal={Discrete Comput. Geom.},
      volume={19},
      number={3, Special Issue},
       pages={447\ndash 455},
        note={Dedicated to the memory of Paul Erd\H os},
      review={\MR{98k:52046}},
}

\end{biblist}
\end{bibdiv}

\end{document}